%% file: stoShortestPath_v11_arxiv.tex
\documentclass[11pt]{amsart}

\usepackage[english]{babel}
\usepackage{graphicx, amsfonts,amsmath,amssymb,amsthm}
\usepackage{fullpage}
\usepackage{comment}
\usepackage[authoryear,sort]{natbib}
\usepackage{algorithm}
\usepackage{algorithmic}
\usepackage{supertabular}
\usepackage{color}

\newenvironment{probleme}[1]{\vskip 0.2 cm\noindent{\bf {\sc
#1}}}{\vspace*{0.2 cm}}
\newcommand{\instanceProblem}{\newline\noindent{\bf Instance.~}}

\newcommand{\solutionProblem}{\newline\noindent{\bf Solution.~}}
\newcommand{\measureProblem}{\newline\noindent{\bf Measure.~}}

\newenvironment{outdent}
{\begin{list}{}{\leftmargin-2cm\rightmargin\leftmargin}\centering\item\relax}
{\end{list}\ignorespacesafterend}

\theoremstyle{plain}
\newtheorem{theo}{Theorem}
\newtheorem{prop}[theo]{Proposition}

\theoremstyle{remark}

\newtheorem{re}{Remark}

\def\Z{\mathbb{Z}}
\def\P{\mathbb{P}}
\def\E{\mathbb{E}}
\def\riskMeas{\rho}
\def\SSPP{\textsc{Stochastic $\riskMeas$-Shortest Path Problem}}
\def\SRCSPP{\textsc{Stochastic $\riskMeas$-Resource Constrained Shortest Path Problem}}
\def\SOTAP{\textsc{Stochastic On Time Arrival Problem}}

\newcommand{\lab}{\lambda}
\newcommand{\labB}{\mu}
\newcommand{\labList}{L}

\newcommand{\NP}{$\textup{NP}$}
\newcommand{\leqst}{\leq_{st}} 
\newcommand{\nleqst}{\nleq_{st}}
\newcommand{\meet}{\wedge} 
\newcommand{\bigmeet}{\bigwedge}
\newcommand{\join}{\vee} 

\author{Axel Parmentier}

\address{A. Parmentier, Universit\'e Paris Est, CERMICS, 6-8 avenue Blaise Pascal, Cit\'e Descartes, 77455 Marne-la-Vall\'ee, Cedex 2, France}
\email{axel.parmentier@cermics.enpc.fr}

\author{Fr\'ed\'eric Meunier}

\address{F. Meunier, Universit\'e Paris Est, CERMICS, 6-8 avenue Blaise Pascal, Cit\'e Descartes, 77455 Marne-la-Vall\'ee, Cedex 2, France}
\email{frederic.meunier@enpc.fr}

\title{Stochastic shortest paths and risk measures}


\begin{document}

\begin{abstract}

We consider three shortest path problems in directed graphs with random arc lengths.

For the first and the second problems, a risk measure is involved. While the first problem consists in finding a path minimizing this risk measure, the second one consists in finding a path minimizing a deterministic cost, while satisfying a constraint on the risk measure. We propose algorithms solving these problems for a wide range of risk measures, which includes among several others the $CVaR$ and the probability of being late. Their performances are evaluated through experiments.

One of the key elements in these algorithms is the use of stochastic lower bounds that allow to discard partial solutions. Good stochastic lower bounds are provided by the so-called \SOTAP{}. This latter problem is the third one studied in this paper and we propose a new and very efficient algorithm solving it.

Complementary discussions on the complexity of the problems are also provided.

\end{abstract}

\maketitle

\section{Introduction}

\subsection{Context}Finding shortest paths in networks has numerous applications. There are direct applications, e.g. organizing the routes of a fleet in logistics or finding the best trajectory for an airplane, but there are also less obvious applications, such as the ones arising as subproblems in column generation approaches.

As for other optimization problems, there is currently a need for taking into account uncertainty in shortest path problems, both from theoretical and practical point of views. Taking into account the uncertainties allows a better optimization. For instance, congestion in road networks has a strong influence on travel times, and modeling delay is crucial when on time arrival is required. Several works have already been carried out on this topic -- see our literature review in Section~\ref{sec:literature} -- but as far as we know, none deal with the problems expressed in the general form we address in the present paper.

\subsection{Problems}We consider various shortest path problems with random arc lengths. They all have in common the following elements: a directed graph $D=(V,A)$ without loops, two vertices $o,d\in V$ that are respectively the origin and the destination, and independent random variables $(X_a)_{a\in A}$ that model random travel times.  The objective is always to find the best $o$-$d$ path. `Best' can take various meanings and we adopt a versatile approach.

The most natural case consists probably in taking no additional constraints into account. In order to be as general as possible, we study the problem of minimizing a given {\em risk measure} of the sum of the $X_a$'s along the path. A risk measure is a mapping from a set of random variables to the set of real numbers. The mean, the variance, or the probability of being above some threshold are examples of risk measures. Another risk measure is the conditional value at risk and is well known in finance but maybe less in operations research. It has many interesting features. Details are given later in the paper. 

Given a risk measure $\riskMeas$, we define the following problem.
 
\begin{probleme}{}\SSPP
\instanceProblem{A directed graph $D = (V,A)$, two vertices $o,d\in V$, independent random variables $(X_a)_{a\in A}$.}
\solutionProblem{An elementary $o$-$d$ path $P$.}
\measureProblem{The quantity $\riskMeas(\sum_{a\in P}X_a)$.}
\end{probleme}

\smallskip


When the $X_a$ are deterministic and $\riskMeas$ is linear, we obtain the usual (deterministic) shortest path problem.\\

Another way to take into account stochasticity when modeling an optimization problem consists in introducing probabilistic constraints. Still assuming given a risk measure $\riskMeas$, we can define the following problem.

\begin{probleme}{}\SRCSPP
\instanceProblem{A directed graph $D = (V,A)$, two vertices $o,d\in V$,  independent random variables $(X_a)_{a\in A}$, a nonnegative number $\riskMeas_0$, and nonnegative numbers $(c_a)_{a\in A}$.}
\solutionProblem{An elementary  $o$-$d$ path $P$ such that
$\riskMeas(\sum_{a\in P}X_a)\leq \riskMeas_0$.}
\measureProblem{The quantity $\sum_{a\in P}c_a$.}
\end{probleme}

\smallskip


When the $X_a$ are deterministic and $\riskMeas$ is linear, we obtain the usual (deterministic) resource constrained shortest path problem.

\subsection{Assumption on the travel times $X_a$}\label{subsec:assump}

All random variables $X_a$, which model the time needed to traverse an arc $a$, are assumed to take their values in $\Z_+$ and to have finite supports. In a computational perspective, it is a natural assumption. When complexity features are discussed, the variable $X_a$ is described by its distribution, which is assumed to be encoded as its support and the probability of each value in the support. 

\subsection{Contributions}

We propose an algorithm for each of the aforementioned problems when the arc lengths $X_a$ satisfy the assumption above and when the risk measure $\riskMeas$ penalizes delay. The mean, the probability of being above some threshold, and the conditional value at risk are examples of such risk measures, but not the variance. ``Penalizing delay'' will be mathematically specified with the help of the so-called usual stochastic order, which is an order defined on the set of random variables -- see Section~\ref{sec:model}. However, our approach remains completely natural. The efficiency of our algorithms, which turns out to be high, is proved by an extensive experimental study. 

One of the main ideas used in these algorithms consists in attaching to each vertex $v$ of the graph a random variable $Z_v$ that is a lower bound of the length of a $v$-$d$ path. ``Lower bound'' is here to be understood in the sense of  the usual stochastic order. As usual with lower bounds in shortest paths algorithms, these stochastic lower bounds allow to discard a priori some partial paths and to reduce the space of possible solutions.

 It turns out that good lower bounds are solutions of another problem, namely the \SOTAP. There are already efficient algorithms solving this problem, but we propose a new one, which seems to be very efficient.
 
 The complexity of the problems dealt with in the paper are also discussed.

\subsection{Plan}

We start with a literature review (Section~\ref{sec:literature}). We describe in Section~\ref{sec:model} the usual stochastic order and the various risk measures covered by our work. The stochastic lower bounds used in our algorithms and their links with the \SOTAP{} are described in Section~\ref{sec:lower_bound}, as well as the new algorithm we propose for this latter problem. The algorithms for the \SSPP{} and the \SRCSPP{} are described in Sections~\ref{sec:SSPP} and~\ref{sec:SRCSPP} respectively, as well as their complexity analysis. The paper ends with an extensive experimental study (Section~\ref{sec:experiments}).

\section{Literature review}\label{sec:literature}

\subsection{Deterministic shortest path}

The shortest path problem with deterministic travel times is one of the most studied problems in Operations Research. Efficient polynomial algorithms are known since the end of the fifties. Recently, research has focused on improving algorithm efficiency thanks to pre-processing, and state of the art techniques compute shortest paths on continental road networks in fractions of a microsecond.  An exhaustive survey on shortest path algorithms for route planning can be found in~\citep{bast2014route}.

\subsection{Deterministic constrained and multicriteria shortest path}

Deterministic shortest paths with resource constraints have many direct applications in logistics and transportation. Moreover, many operations research problems  require to find such shortest paths as a subroutine (e.g. vehicle routing problems). Resource constrained shortest path have therefore been extensively studied, and a detailed literature review is out of the scope of this paper. A review of state of the art techniques can be found in \citep{irnich2005shortest}. 

Several problems close to the resource constrained shortest path problem belong to multicriteria optimization \citep{ehrgott2005multicriteria}. Many algorithms have been developed for the bicriteria \citep{raith2009comparison} and multicriteria shortest path problem \citep{tarapata2007selected}. State of the art solution methods for resource constrained or multicriteria shortest path problems often rely on labeling algorithms with pre-processing \citep{dumitrescu2003improved,boland2006accelerated}. 

\subsection{Stochastic shortest path}

Stochastic shortest path problems have been extensively studied since the seminal work of \cite{frank1969shortest}. Models differ by the probability distributions used to model delay on arcs, and by the risk measure optimized.

A first line of papers considers the probability of on time arrival: a path maximizing the probability of on time arrival, or analogously, a path with minimum quantile of given order is searched. Approaches have been developed for both continuous \citep{frank1969shortest,chen2005path,nikolova2006stochastic,nikolova2010high} and discrete distributions \citep{mirchandani1976shortest}. An efficient labeling algorithm has recently been described when arc distributions are normal \citep{chen2013finding}. In a recent preprint, \cite{niknami2014tractable} proposes independently an idea that bears similarities with ours by using the solution of the \SOTAP{} to obtain lower bounds for the \SSPP{} with $\riskMeas(\cdot)=\P(\cdot\geq\tau)$. Note that our approach covers actually more general situations, even in this special case.


A second line of papers defines a shortest path as a path minimizing the expectation of a cost function \citep{loui1983optimal}.
Dynamic programming can be used when cost functions are affine or exponential \citep{eiger1985path}.
\cite{murthy1996relaxation,murthy1998stochastic} present an efficient labeling algorithm when arc distributions are normal and cost functions are piecewise-linear and concave. Instead of considering the expectation of a cost functions, \cite{sivakumar1994variance,nikolova2006stochastic,nikolova2010high} search the path that minimizes a positive linear combination of mean and variance.

Finally, \cite{miller2003path} suggest to use stochastic dominance to compare paths. Algorithms to generate all non-dominated paths are proposed in \cite{miller1997optimal,miller1998least,miller2003path,nie2009shortest,nie2012optimal}.

\subsection{Stochastic on time arrival}
The \textsc{Stochastic on Time Arrival Problem} searches to maximize the probability of arrival before a given thresholds on adaptive paths \citep{fu2001adaptive,fu1998expected,hall1986fastest}. An adaptive path is an application which, given a vertex reached and the time it took to get to this vertex indicates the next arc to choose to maximize the probability of on time arrival at destination. \cite{fan2006optimal} provide an algorithm for continuous distributions, and \cite{nie2006arriving} provide a pseudo-polynomial algorithm for discrete distributions. \cite{samaranayake2012tractable} develop a faster algorithm for discrete distributions, and \cite{sabran2014precomputation} provide pre-processing techniques to improve algorithm speed. 

\subsection{Shortest path under probability constraint}

Finally, the problem of finding a minimum cost path for deterministic arc costs under stochastic resource constraints have been introduced in \cite{kosuch2010stochastic}, which proposes a solution algorithm based on linear programming is derived.

\section{Usual stochastic order and risk measures}\label{sec:model}
\subsection{Usual stochastic order}

For any random variable $X$, let $F_X(t)=\P(X\leq t)$ be its cumulative distribution. The {\em usual stochastic order} $\leqst$ is a partial order defined on the set of cumulative distributions as follows: $F_X\leqst F_Y$ if $F_X(t) \geq F_Y(t)$ for all $t$. Endowed with the usual stochastic order, the set of cumulative distributions turns out to be a {\em lattice}.  It means that each collection of cumulative distributions has a unique least upper bound ({\em join}) and a unique greatest lower bound ({\em meet}). \\

We use the same notation $\leqst$ to extend the usual stochastic order to random variables. Two random variables $X$ and $Y$ are such that $X \leqst Y$ if $F_X \leqst F_Y$.  Note that in the context of random variables, it is a {\em quasiorder}: it is reflexive and transitive like a partial order, but not antisymmetric. The usual stochastic order is a classical notion in probability theory and enjoys several properties. Missing details and a more complete introduction may be found in~\citet{muller2002comparison}.  \\

One property is particularly useful in our work and is used in the paper without further mention.\\

{\em Let $X_{1}$, $X_{2}$, $Y_{1}$, and $Y_{2}$ be random variables such that $X_1$ and $Y_1$ are independent, and $X_2$ and $Y_2$ are independent. If $X_{1} \leqst X_{2}$ and $Y_{1} \leqst Y_{2}$, then $X_{1} + Y_{1} \leqst X_{2} + Y_{2}$.} \\

By a slight abuse, we define the meet (resp. the join) of a collection $(X_i)$ of random variables as a random variable with cumulative distribution equal to the meet (resp. the join) of the $F_{X_i}$'s. We use the classical notation $\meet$ for the meet and $\join$ for the join. We have in particular for any collection $(X_i)$ of random variables the following equality $$F_{\meet_i X_i}(t) = \max_i(F_{X_i}(t))\quad\mbox{for all $t\in\Z_+$.}$$

The meet of the empty collection is as usual equal to $+\infty$ and its join is equal to $-\infty$.

Since the supports of the random variables are assumed to be integer valued and of finite cardinality, their meet (resp. the join) can easily be computed: for each $t$ in the union of their supports, we simply keep the largest (resp. smallest) value among those taken by the cumulative distributions at $t$. In the paper, we only use the meet operation.

\subsection{Risk measures}

A risk measure $\riskMeas$ is said to be \textit{consistent with the usual stochastic order} if $X \leqst Y$ implies $\riskMeas(X)\leqst \riskMeas(Y)$. The algorithms for the \SSPP{} and the \SRCSPP{} introduced in this paper only require the consistency of the risk measure with the usual stochastic order. It is not a strong assumption since being consistent with the risk measure simply means penalizing delay. Note however that the variance does not satisfy this condition.

For examples, the following risk measures are covered by our work.
\begin{itemize}
\item[$\bullet$] $\riskMeas(\cdot)=\E[f(\cdot)]$ for any increasing function $f$.
\item[$\bullet$] $\riskMeas(\cdot)=\P(\cdot\geq\tau)$ for any $\tau\in\Z_+$.
\item[$\bullet$] $\riskMeas(\cdot) =VaR_\beta(\cdot)$ for any $\beta\in[0,1]$.
\item[$\bullet$] $\riskMeas(\cdot)=CVaR_\beta(\cdot)$ for any $\beta\in[0,1]$.
\end{itemize}

The first risk measure among the ones listed can be used for instance when the cost of being late is modeled. We can cite as an illustration the case of goods delivery, an increasing step function can model the successive penalties of a late delivery. 

The second one is typically used to measure the probability of going beyond some deadline. It corresponds to the risk measures of a traveler going to the airport: the only onjective is to arrive before the time $\tau$ when boarding gate is closed.

The third and fourth risk measures are widely used in finance.

$VaR_\beta(\cdot)$ is the {\em value at risk (with confidence level $\beta$)} and is defined by 
$$VaR_\beta(X)=\min\{t \in\Z_+:\, \P(X\leq t) \geq \beta\}.$$ Given a confidence level $\beta$, a random variable $X$ is smaller than its value at risk $VaR_\beta(X)$ with probability at least $\beta$, and $VaR_\beta(X)$ is the threshold value for which this inequality holds. For instance, if a bus company advertises the policy ``our buses are on time for $\beta \%$ of the trips'', then the shortest path with respect to $VAR_\beta$ gives the smallest travel time for which this policy holds.

$CVaR_\beta(\cdot)$ is the {\em conditional value at risk (with confidence level $\beta$)} and is defined by $$CVaR_\beta(X)=\E\left[X | X \geq VaR_\beta(X) \right].$$ Given a confidence level $\beta$, the conditional value at risk of a random variable $X$ can be described in a somehow informal way as its average value in the $1- \beta$ worst cases. In the context of shortest paths with random travel times, a $\beta$ equal to $0$ corresponds to a shortest path problem where expectation is minimized. A $\beta$ equal to $1$ corresponds to a robust approach: the optimal solution is the solution with minimal worst case behavior. As a consequence, $CVaR_{\beta}$ enables to obtain a tradeoff between expectation and robustness. This cost function suits well most traveler behavior: a small path with small $CVaR$ gives a small average traveling time with a low risk of very large travel time. 

\section{Stochastic lower bounds and On Time Arrival}\label{sec:lower_bound}

\subsection{Stochastic lower bounds}

Efficient shortest path algorithms are often based on lower bounds that allow to discard partial paths and to reduce the space of solutions that is explored. We do not depart from it and introduce stochastic lower bounds, which, instead of being single values, are random variables. 

We consider the following equation, where $\delta^+(v)$ denotes the set of arcs of the form $(v,u)$.
\begin{equation}\label{eq:stoc_lower}
\left\{\begin{array}{ll}Z_{d} = 0, &\\
Z_v \stackrel{(d)}{=} \bigmeet_{(v,u)\in\delta^+(v)} \left(X_{(v,u)}+Z_u\right) & \mbox{for all $v\in V\setminus\{d\}$}.
\end{array}\right.
\end{equation}
Proposition~\ref{prop:SOTAopt} below shows that this equation has always a solution which can be calculated in polynomial time. Before stating and proving this proposition, we explain the interest of this equation for stochastic shortest path problems.

\begin{prop}\label{prop:stoc_lower} Let $(Z_v)_{v\in V}$ be a solution of Equation~\eqref{eq:stoc_lower}.
For any vertex $v\in V$ and any $v$-$d$ path $P$, the inequality $Z_v\leq_{st}\sum_{a\in P}X_a$ holds.
\end{prop}
\begin{proof}
An induction on the number of arcs in $P$ allows to conclude.
\end{proof}

\subsection{On Time Arrival}

The \SOTAP{} aims at finding the path in $D$ that maximizes the probability of reaching $d$ before some time limit $\tau\in\Z_+$. If the solution has to be computed a priori, we are exactly in the case of the \SSPP, with $\riskMeas(X)=\P(X\geq\tau)$, which is a risk measure covered by the present work. In the \SOTAP, the path can be built during the trip. Concretely, an optimal policy has to computed. A policy $\pi$ for this problem associates to each pair $(v,t)\in V\times\Z_+$ an arc $a\in\delta^+(v)$. Such a policy $\pi$ is optimal if it maximizes the probability of reaching $d$ before $\tau$. If there is in $D$ a circuit $C$ with $\P(X_a=0)=1$ for all $a\in C$, the problem may not be well-defined. We suppose therefore that there are no such circuits. Under this assumption, and using elementary properties of Markov chains, we get that solving \SOTAP{} is equivalent to finding a solution of the following equation for all $t\in\Z_+$ (see also \cite{fan2005arriving} for a first formulation of this kind in a continuous setting):


\begin{equation}\label{eq:stoc_ontime}
\left\{\begin{array}{ll}
F_{d}(t) = 1 & \\ \\
\displaystyle{F_v(t) = \max_{(v,u)\in\delta^{+}(v)} \sum_{k=0}^t \P(X_{(v,u)}=k) F_u(t-k)} &\mbox{for all $v \in V\setminus\{d\}.$}
\end{array}\right.
\end{equation}

The quantity $F_v(t)$ is then the probability to reach $d$ from $v$ in less than time $t$ under an optimal policy. The decision to be taken at time $0$ on vertex $o$ is determined by the $(o,u)\in A$ for which the maximum is attained when $v=o$ and $t=\tau$ in Equation~\eqref{eq:stoc_ontime}.

Given a realization of the random variables $X_a$, the path obtained following the policy of Equation~\eqref{eq:stoc_ontime} is not necessary elementary. \citet{samaranayake2012tractable} provides a simple example of network in which an event realization leads to a non-simple path. We therefore adopt the following hypothesis: when a path crosses several times the same arc $a$, then we consider that delays encountered on the arc correspond to independent realizations of the random variable $X_a$.

Without loss of generality, we can assume that there exists in $D$ at least one $v$-$d$ path (otherwise, the vertex $v$ can simply be removed from $D$). We have $\lim_{t\rightarrow+\infty}F_v(t)=1$ for all $v\in V$ and $F_{v}(\cdot)$ is non-decreasing. Thus, $F_{v}(\cdot)$ can be considered as a cumulative distribution, and we can introduce a random variable $U_v$ whose cumulative distribution function is $F_v(\cdot)$.

This construction shows that any algorithm solving Equation~\eqref{eq:stoc_lower} solves the \SOTAP{} as well.

\begin{prop}\label{prop:lower_stoc_online}
The $U_v$'s defined above satisfy Equation~\eqref{eq:stoc_lower} and conversely, any solution of Equation~\eqref{eq:stoc_lower} has cumulative distributions that satisfy Equation~\eqref{eq:stoc_ontime}.
\end{prop}
\begin{proof}
The equivalence between \eqref{eq:stoc_lower} and \eqref{eq:stoc_ontime} follows directly from the equality $$F_{X \meet Y}(t) = \max(F_{X}(t),F_{Y}(t)).$$
\end{proof}

\begin{re}
If all $X_a$'s are deterministic, the \SOTAP{} coincides with the usual deterministic shortest problem, for which by the way the solution computed a priori coincides with the solution computed during the trip.
\end{re}

\subsection{An algorithm}

We describe now a new polynomial algorithm computing a solution of Equation~\eqref{eq:stoc_lower}, and thus solving also the \SOTAP.

A random variable $Z'_v$ and an integer $t'_v$ are attached to each vertex $v$ and updated during the algorithm. Initially, $Z'_v=+\infty$ and $t'_v=+\infty$ for each vertex $v\neq d$, while $Z'_d=0$ and $t'_d=0$.
During the algorithm, a queue $L$ of vertices ``to be expanded'' is maintained. Initially, the queue $L$ contains only $d$.

The algorithm ends when $L$ is empty. While $L$ is not empty, the following operations are repeated:
\begin{itemize}
\item Extract from $L$ the vertex $u$ with minimum $t'_u$. In case there are several such vertices, choose one with maximal value of $F_{Z'_u}(t'_u)$.
\item Set $t'_{u} = +\infty$.
\item For each arc $(v,u)$ in $\delta^-(u)$, {\em expand} $u$ along $(v,u)$: if $Z'_v \nleqst X_{(v,u)}+Z'_{u}$, then
	\begin{itemize}
	\item {\em Update} $t'_{v} = \min(t'_{v}, \min\{t:\, F_{Z'_{v} }(t) < F_{X_{(v,u)}+Z'_{u}}(t)\})$.
	\item {\em Update} $Z'_{v}$ to $Z'_{v} \meet (X_{(v,u)}+Z'_{u})$.
	\item Add $v$ to $L$ (if it is not already present).
	\end{itemize}
\end{itemize}

\begin{prop}\label{prop:SOTAopt}
This algorithm terminates in less than $T|V|$ iterations, where $T$ is the maximum length an elementary $o$-$d$ path can take (over all possible paths and over all possible events). Setting $Z_v$ to be the last value of $Z'_v$ for each $v\in V$ provides then a solution of Equation~\eqref{eq:stoc_lower}. 
\end{prop}

Algorithms in \cite{nie2006arriving} and in~\cite{samaranayake2012tractable} make a stronger assumption on arc distributions: they suppose the existence of a $\delta>0$ such that $X_{a} \geq \delta$ for all arcs $a$. Moreover, the number of iterations of the algorithm in \citet{nie2006arriving} is always $T|V|$ and in \citet{samaranayake2012tractable} it is a $\Omega(T|V|)$, while in our approach, it corresponds to a worst case behavior: In general, we expect our algorithm to terminate after $\gamma |V|$ iterations for a small $\gamma$. The worst $\gamma$ encountered in the numerical experiments was $3.3$, see Section \ref{sec:experiments} for more details. 

\begin{proof}[Proof of Proposition \ref{prop:SOTAopt}]
At any time during the algorithm, let $Z''_v$ have the distribution of $Z'_v$ the last time $v$ was expanded and let $t''_v$ be the value of $t'_v$ right before being set to $+\infty$ because of the expansion.
If $v$ has never been expanded, let $Z''_v=+\infty$ and $t''_v=-1$. The update rule of $t'_v$ ensures thus that $$F_{Z''_v}(t)=F_{Z'_v}(t)\quad\mbox{for all $t\leq t'_v-1$}.$$

Thus, if the algorithm terminates, we have $t'_v=+\infty$ and $Z''_v=Z'_v$ for each vertex $v$. As $Z'_v$ is only updated when one of its outneighbor $u$ is expanded, we have 
\begin{equation}\label{eq:in_algo}
Z'_v \stackrel{(d)}{=} \bigmeet_{(v,u)\in\delta^+(v)} \left(X_{(v,u)}+Z''_u\right).
\end{equation}
Therefore $Z'_v$ is then a solution to Equation~\eqref{eq:stoc_lower}. 

It remains to prove that the algorithm terminates in a finite number of iterations. \\

We prove that $\min_{v\in V}t'_v$ does not decrease along the algorithm, and that if this quantity remains constant for consecutive iterations, then $\max_{v\in V}F_{Z'_v}(t'_v)$ does not increase.

Equation~\eqref{eq:in_algo} implies that $Z'_v\leqst X_{(v,u)}+Z''_u$ for any $(v,u)\in\delta^+(v)$. Since $F_{Z'_u}$ and $F_{Z''_u}$ coincide up to $t'_u-1$, we have 
\begin{equation}\label{eq:ineq_distr}F_{Z'_v}(t)\geq F_{X_{(v,u)}+Z'_{u}}(t)\quad\mbox{for all $t\leq t'_u-1$}.\end{equation} Moreover,  just before being updated because of the expansion of a vertex $u$, the integer $t'_v$ is equal to $t''_u$ because of the way the vertex $u$ has been selected in $L$. This remark and Equation~\eqref{eq:in_algo} implies together that after an expansion along an arc $(v,u)$, we have $t'_v\geq t''_u$. 

Consider the algorithm right after a vertex $u$ has been expanded. 

Let $v$ be such that $(v,u)\in\delta^+(v)$.
If $t'_v$ has just been updated to $t''_u$, then $$F_{Z'_v}(t'_v)=\sum_{k=0}^{t''_u}\P(X_{(v,u)}=k)F_{Z'_u}(t''_u-k),$$ and since $F_{Z'_u}$ is a nondecreasing map, we have $F_{Z'_v}(t'_v)\leq F_{Z'_u}(t''_u)$. If $t'_v$ has not been updated while being already such that $t'_v=t''_u$, then $F_{Z'_v}(t'_v)\leq F_{Z'_u}(t''_u)$ because of the selection rule of $u$. 
Therefore, the next vertex $w$ to be expanded will either be such that $t'_w>t''_u$, or such that $t'_w=t''_u$ and $F_{Z'_w}(t'_w)\leq F_{Z'_u}(t''_u)$. \\

To finish the proof of termination, we show now that $t'_v\geq t''_v+1$ at any time during the algorithm. If $t'_v=+\infty$, the inequality is clearly satisfied. Otherwise, let $u$ be the last vertex that has been expanded and whose expansion has lead to an update of $t'_v$. We consider the algorithm right after this expansion. $t''_u$ is thus the value of $t'_u$ right before having being set to $+\infty$, and $t'_v$ has been updated by the expansion.

Note that $t''_u\geq t''_v$ since $\min_{w\in V}t''_w$ does not decrease along the algorithm (it does not decrease because each $t''_w$ is the value of $\min_{v'\in V} t'_{v'}$ at some previous iteration, and this latter quantity is non decreasing, as already noted).  We already know that $t'_v\geq t''_u$. Thus, if $t'_v>t''_u$, we have $t'_v\geq t''_v+1$ as required. It remains to check whether it is possible to have $t'_v=t''_u=t''_v$ simultaneously. In such a case, we would necessarily have $F_{Z'_v}(t''_v)\geq F_{Z'_u}(t''_u)$ since $\max_{w\in V}F_{Z'_w}(t'_w)$  does not increase  when $\min_{w\in V}t'_w$ remains constant. Since $F_{Z'_u}$ is a nondecreasing map, we would have $F_{Z'_u}(t''_u)\geq F_{X_{(v,u)}+Z'_{u}}(t''_u)$ and thus $$F_{Z'_v}(t)\geq F_{X_{(v,u)}+Z'_{u}}(t)\quad\mbox{for all $t\leq t''_u$}$$ with the help of Equation~\eqref{eq:ineq_distr}. Thus, $t'_v$ should have been updated to a value greater than $t''_u$. Hence, in any case, we have $t'_v\geq t''_v+1$.\\

Now, denote by $T_v$ the maximum length an elementary $v$-$d$ path can take. We claim that once $Z'_v$ has been updated for the first time, we have $F_{Z'_v}(T_v)=1$ all along the algorithm. It can be proved by a direct induction on the number of arcs an elementary path may have with the help of Equation~\eqref{eq:in_algo}. Moreover, since the inequality $t'_v\geq t''_v+1$ holds along the algorithm, we may have $F_{Z'_{v} }(t) < F_{X_{(v,u)}+Z'_{u}}(t)$ only for $t>t''_v$. Therefore, after at most $T_v$ iterations, the test ``$Z'_v\nleqst X_{(v,u)}+Z'_{u}$'' is always false. It implies that after at most $\max_{v\in V}T_v$ iterations, the algorithm terminates.
\end{proof}

\begin{re}
For each elementary path $P$, define $T_P$ to be the maximum length it can takes over all possible realizations. Denote by $\mathcal{P}_{vd}$ the set of all elementary $v$-$d$ paths.
Proposition~\ref{prop:SOTAopt} is actually true for $T=\min_{P\in\mathcal{P}_{od}}T_P$. Indeed, in the proof above, we show that the algorithms terminates after a finite number of iterations. Denote $\overline{T}_v=\min_{P\in\mathcal{P}_{vd}}T_P$. Proposition~\ref{prop:stoc_lower} ensures that $F_{Z_v}(\overline{T}_v)=1$. Thus, as soon as $t''_v=\overline{T}_v$, the test ``$Z'_v\nleqst X_{(v,u)}+Z'_{u}$'' is always false. It implies that after at most $\max_{v\in V}\overline{T}_v$ iterations, the algorithm terminates.
\end{re}


\begin{re}
When all $X_a$ are deterministic, this algorithm coincides with Dijkstra's algorithm.
\end{re}

\section{\SSPP}\label{sec:SSPP}

\subsection{Complexity}

When $\riskMeas$ is linear over independent random variables and dealt with as an oracle, the \SSPP{} can clearly be solved in polynomial time by any Dijsktra-like algorithm.  
However, we were not able to decide whether the problem remains polynomial in the general case, still dealing with $\riskMeas$ as an oracle. Nevertheless, we are able to prove the following theorem. 

\begin{theo}\label{theo:SSPPcomplexity} There is no polynomial algorithms with a complexity function independent of $\riskMeas$ solving 
the \SSPP, unless $\textup{P}=$\NP.
\end{theo}

\begin{proof}The proof consists in describing a polynomial reduction of the deterministic resource constraint path problem to some \textsc{Stochastic $\overline\riskMeas$-Shortest Path Problem}, with an adequate risk measure $\overline\riskMeas$ consistent with the usual stochastic order. The deterministic resource constrained path problem consists in a directed graph $D=(V,A)$, two vertices $o,d\in V$, nonnegative integers $(c_a)_{a\in A}$ (the costs), nonnegative integers $(r_a)_{a\in A}$ (the resources), and a nonnegative integer $R$ (the capacity). It aims at finding a path $P$ with $\sum_{a\in P}r_a\leq R$ and with minimal cost $\sum_{a\in P}c_a$.
It is an \NP-hard problem, see~\citet{handler1980dual}.

To that purpose, we define two risk measures $\riskMeas_{\min}$ and $\riskMeas_{\max}$.
$$\begin{array}{rcl}
\riskMeas_{\min}(X) & = & \min\{t\in\Z_+:\, F_X(t)>0)\} \\ \\
\riskMeas_{\max}(X) & = & \left\{
\begin{array}{rl}
0 &\text{ if } \max\{t\in\Z_+:\,\P(X=t)> 0 \}\leq M  \\
1 & \text{ otherwise,}
\end{array}
\right.
\end{array}
$$where $M=R(1+\max_{a\in A}c_a)$. We define
$$\overline\riskMeas=\riskMeas_{\min}+\left(\sum_{a\in A}c_a\right)\riskMeas_{\max}.$$
The risk measure $\riskMeas_{\min}$ is consistent with the usual stochastic order because $F_{X}(t) \geq F_{Y}(t)$ for all $t$ implies $ \min\{t\in\Z_+:\, F_X(t)>0)\} \leq \min\{t\in\Z_+:\, F_Y(t)>0)\}$. As well as $\rho_{min}$, the risk measure $\riskMeas_{\max}$ is consistent with the usual stochastic order because $\max\{t\in\Z_+:\,\P(X=t)> 0 \} = \min\{t\in\Z_+:\,F_{X}(t)=1 \}$ and $F_{X}(t) \geq F_{Y}(t)$ for all $t$ implies $ \min\{t\in\Z_+:\, F_X(t)=1)\} \leq \min\{t\in\Z_+:\, F_Y(t)=1)\}$.  Thus $\overline\riskMeas$ is also consistent with the usual stochastic order.

We describe now the reduction to the \textsc{Stochastic $\overline\riskMeas$-Shortest Path Problem}. Given an arc $a$, we define $X_a$ as follows: 
$$\P(X_a = c_a) = \frac{1}{2}\quad\mbox{ and }\quad\P\left(X_a = r_a(1+\max_{b \in A} c_b)\right) = \frac{1}{2}.$$ $X_a$ can only take two values.

The deterministic resource constrained shortest path problem has a feasible solution if and only if the \textsc{Stochastic $\overline\riskMeas$-Shortest Path Problem} has a solution of cost less than $\sum_{a\in A}c_a$. In this case, the optimal solutions of both problems coincide, as we show now.

We assume that the deterministic resource constrained shortest path problem has a feasible solution. Let $P$ be  an optimal solution of the \textsc{Stochastic $\overline\riskMeas$-Shortest Path Problem}. We have then $\overline\riskMeas(\sum_{a\in P}X_a) =\sum_{a\in P} c_a$ and $\sum_{a\in P} r_a \leq R$, which implies that the deterministic resource constraint shortest path problem has a feasible solution of cost $\sum_{a\in P} c_a$.  Conversely, an optimal solution $P'$ of the deterministic resource constrained shortest path problem provides a feasible solution of the \textsc{Stochastic $\overline\riskMeas$-Shortest Path Problem} of cost $\sum_{a\in P'} c_a$.
\end{proof}

Note the $\overline\riskMeas$ used in the proof being fixed, the \SSPP{} becomes polynomially solvable. The proof above works precisely because the complexity function must be independent of the risk measure.


\subsection{Algorithm}\label{subsec:algo_ssp}

For any path $P$, we denote by $X_P$ the random variable $\sum_{a\in P}X_a$. Sometimes in the proofs or in the computations, the path $P$ can be non-simple. In this case, independent copies of the $X_a$'s appearing several times in the sum are used. We assume given for each vertex $v$ a random variable $Z_v^{LB}$ such that $Z_v^{LB} \leqst X_P$ for all $v$-$d$ paths $P$. Such a collection of random variables $(Z_v^{LB})_{v\in V}$, which play the role of stochastic lower bounds, can be computed for instance by the techniques presented in Section~\ref{sec:lower_bound}.

These stochastic lower bounds are used to discard some partial paths that are not subpaths of an optimal path. It resembles classical techniques used by shortest path algorithms in the deterministic setting, see for instance \cite{bast2014route} for state of the art techniques, in which the lower bound is a real number, while in our setting the lower bound is a random variable.  \\

We present the algorithm, which is a labeling one. A {\em label} $\lab$ is a pair $(v_{\lab},Y_{\lab})$, where $v_{\lab}$ is a vertex and $Y_{\lab}$ is a random variable. Labels are stored in a queue $\labList$. Initially, $\labList$ contains a unique label $\lab_{o} = (o,0)$. If there is no $o$-$d$ paths, the algorithm returns nothing. Otherwise, an elementary $o$-$d$ path $P_0$ is computed and $\riskMeas_{od}^{UB}$ is set to $\riskMeas(X_{P_0})$. The algorithm ends when $\labList$ is empty. While $\labList$ is not empty, the following operations are repeated: 
\begin{itemize}
\item[$\bullet$] Extract a label $\lab$ of $\labList$. 
\item[$\bullet$] If $v_{\lab}=d$ and $\riskMeas(Y_{\lab})<\riskMeas_{od}^{UB}$: update $\riskMeas_{od}^{UB}$ to $\riskMeas(Y_{\lab})$.
\item[$\bullet$] Otherwise:
If $\riskMeas(Y_{\lab}+Z_{v_{\lab}}^{LB})<\riskMeas_{od}^{UB}$, {\em expand} label $\lab$: for each arc $(v_{\lab},v)$ in $\delta^+(v_{\lab})$, add a new label $(v,Y )$ to $\labList$, where $Y\stackrel{(d)}{=}Y_{\lab} + X_{(v_{\lab},v)}$.
\end{itemize}

\begin{theo}\label{theo:stoBackward} Suppose that $\P(X_a\neq 0)> 0$ for all $a\in A$ and that there exists at least one solution. If $\riskMeas$ is consistent with the usual stochastic order, then the algorithm described above terminates after a finite number of iterations and at the end, $\riskMeas_{od}^{UB}$ is the optimal value of a solution of the \SSPP.
\end{theo}

Before giving the proof, we give the two ideas the algorithm relies on. Note that a label $\lab$ corresponds actually to some $o$-$v_{\lab}$ path $P_{\lab}$ that the algorithm tries to expand in an optimal path (more details are given in the proof below).

The first idea is the following. At any time during the algorithm, $\riskMeas_{od}^{UB}$ is an upper bound on the optimal cost. Suppose that $\riskMeas_{od}^{UB}$ is not tight. If $P$ is an $o$-$v$ subpath of an optimal solution, then
\begin{equation}\label{eq:pathDomDiscard}
\riskMeas(X_P+Z_v^{LB}) <\riskMeas_{od}^{UB}.
\end{equation}
In an enumeration of all the $o$-$d$ paths to find an optimal path, any partial path $P$ that does not satisfy Equation~\eqref{eq:pathDomDiscard} can thus be discarded.

The second idea is that, when an $o$-$d$ path $P$ is such that $\riskMeas(X_P) < \riskMeas_{od}^{UB}$, then $\riskMeas_{od}^{UB}$ can be updated to $\riskMeas(X_P)$. Such an update reduces the number of paths satisfying Equation~\eqref{eq:pathDomDiscard}. \\

As it is described, the algorithm computes only the optimal cost of a solution. The algorithm can easily be adapted to find the optimal path itself, 
simply by maintaining, for each label $\lab$ the label whose expansion has lead to its creation.  


\begin{re}\label{re:Xa0}
It is easy to deal with the presence of variables $X_a$ with $X_a=0$. If $X_{(u,v)}=0$, remove $(u,v)$ and add arcs $(w,v)$ for $w\in N^-(u)$ with $X_{(w,v)}\stackrel{(d)}{=}X_{(w,u)}$ as well as arcs $(u,w)$ for $w\in N^+(v)$ with $X_{(u,w)}\stackrel{(d)}{=}X_{(v,w)}$.
\end{re}

\begin{proof}[Proof of Theorem~\ref{theo:stoBackward}]
We associate to each label $\lab$ an $o$-$v_{\lab}$ path $P_{\lab}$ such that $Y_{\lab}=X_{P_{\lab}}$: If $P_{\labB}$ is associated to the label $\labB$ and if $\lab$ is obtained by expanding $\labB$, we define $P_{\lab}$ as the path $P_{\labB}+(v_{\labB},v_{\lab})$. It is easy to see that there is at most one label associated to a given path. This remark will be useful in the proof.

Initially, $\riskMeas_{od}^{UB}$ is equal to $\riskMeas(X_{P_{0}})$, where $P_{0}$ is the elementary $o$-$d$ path computed at the beginning of the algorithm. Besides, $\riskMeas_{od}^{UB}$ can only decrease during the algorithm. In addition, we have $\riskMeas(X_{P_{\lab}}) \leq \riskMeas(Y_{\lab} + Z_{v_{\lab}})$ by consistency of $\riskMeas$ with the usual stochastic order, and because $X_{P_{\lab}} \stackrel{(d)}{=} Y_{\lab} \leqst Y_{\lab} + Z_{v_{\lab}}^{LB}$. Thus, as only labels such that $\riskMeas(Y_{\lab}+Z_{v_{\lab}}^{LB}) < \riskMeas_{od}^{UB}$ are expanded, and as, given a path $Q$, there is at most one label such that $P_{\lab} = Q$, the number of labels that are expanded during the algorithm is upper-bounded by the number of paths $P$ such that $\riskMeas(X_{P}) < \riskMeas(X_{P_{0}})$. 

To prove that the number of paths $P$ satisfying $\riskMeas(X_{P}) < \riskMeas(X_{P_{0}})$ is finite, we need the following claim. \\

{\em Let  $X$ and $Y$ be random variables satisfying the assumption of Section~\ref{subsec:assump} and denote by $S_k$ be the sum of $k$ independent copies of $X$. Suppose that $\P(X\neq 0)>0$. Then there exists an integer $n$ such that $Y\leqst S_n$.}\\

This claim is proved as follows. Given an integer $t$, the event $S_{k}\leq t$ requires that at least $k-t$ copies of $X$ are equal to $0$. Using this idea and then the Stirling formula, we obtain $F_{S_k}(t)\leq {k\choose{t}}P(X=0)^{k-t} = O(k^{t}P(X=0)^{k-t})$. For any integer $t$, we have thus $\lim_{k\rightarrow+\infty}F_{S_k}(t)=0$. It implies that for any $t\leq t_{\max}^Y$, where $t_{\max}^Y$ is the maximum value $Y$ can take, there is a $k_t$ such that $F_{S_{k}}(t)<F_Y(t)$ when $k\geq k_t$. The quantity $t_{\max}^Y$ is finite, because  of the finiteness of the support of $Y$. Defining $n=\max_{0\leq t\leq  t_{\max}^Y} k_t$ proves the claim. \\

As a consequence of the claim, there is only a finite number of paths $P$ such that $X_{P_{0}} \nleqst X_{P}$. As $X_{P_{0}} \leqst X_{P}$ implies $\rho(X_{P_{0}}) \leq \rho(X_{P})$ by consistency of $\riskMeas$ with the usual stochastic order, there is only a finite number of paths $P$ such that $\riskMeas(X_{P}) < \riskMeas(X_{P_{0}})$. Hence, only a finite number of labels are expanded and the algorithm terminates after a finite number of iterations.

We prove now that the algorithm returns the correct value. Note that $\riskMeas_{od}^{UB}$ is at any time the cost of some $o$-$d$ path. Let $P^*$ be the $o$-$d$ path that has updated $\riskMeas_{od}^{UB}$ for the last time. Note that $P^*$ is not necessarily elementary.

Suppose that there is at least one $o$-$d$ path, otherwise the algorithm would return $+\infty$. Denote by $P$ any elementary $o$-$d$ path. There are subpaths $Q$ of $P$ starting at $o$ and for which there exists a label $\lab$ such that $Q=P_{\lab}$. The path reduced to the vertex $o$ is one of them. Let $P_1$ be such a subpath with the largest number of arcs, $\labB$ be such that $P_1=P_{\labB}$, and $v$ be the destination of $P_1$. We also denote by $P_2$ the  $v$-$d$ subpath of $P$. Since $P_1$ has been chosen with a largest number of arcs, the label $\labB$ has not been expanded. If $v=d$, we have $\riskMeas(X_P)$ larger than or equal to the last value of $\riskMeas_{od}^{UB}$. Otherwise, $v\neq d$ and we had
$\riskMeas(Y_{\labB}+Z_{v}^{LB})\geq\riskMeas_{od}^{UB}$ when $\labB$ has been extracted. Since $\riskMeas_{od}^{UB}$ is non-increasing during the algorithm and since we have $X_{P_2}\geq_{st} Z_{v}^{LB}$, the cost $\riskMeas(X_P)$ is necessarily larger than or equal to the last value of $\riskMeas_{od}^{UB}$, which is larger or equal to the cost of some elementary $o$-$d$ path: $P^*$ itself if it is elementary, or the elementary subpath obtained by removing the cycles otherwise. Thus the last value of $\riskMeas_{od}^{UB}$, which is equal to $\riskMeas(P^*)$, is the optimal cost.\end{proof}

\subsection{Speeding up the algorithm}\label{subsec:full_algo_stoch}

\subsubsection{Keys}

We can speed up the algorithm by making $\labList$ a priority queue and using a key equal to $\riskMeas(Y_{\lab}+Z_{v_{\lab}}^{LB})$ for the element $(v_{\lab},Y_{\lab})$. The idea behind this choice is that we give priority to the label that seems to be the most promising, as $\riskMeas(Y_{\lab}+Z_{v_{\lab}}^{LB})$ is a lower bound on the cost of any $o$-$d$ having $P_{\lab}$ as subpath. The efficiency of the methods depends on the quality of the stochastic lower bound $Z_{v_{\lab}}^{LB}$. In such a way, we may discard more partial paths, since the upper bound $\riskMeas_{od}^{UB}$ will likely be smaller than without the use of such keys.

\subsubsection{Upper bounds}\label{sec:speedUB}

Suppose that, in addition to the lower bounds $Z_v^{LB}$, we maintain for each vertex $v$ a list $\ell_v=(X_{P_{1}}, X_{P_{2}}, \ldots , X_{P_{k}})$, where the $P_i$ are elementary $v$-$d$ paths. Then, these paths can be used to update the minimum known cost during the algorithm.
Let $\lab$ be a label. The path $P_{\lab}$ is an $o$-$v_{\lab}$ path. As a consequence, for any distribution $X_{P_{i}}\in \ell_{v_{\lab}}$, path $P_{\lab} + P_{i}$ is an $o$-$d$ path. Thus, if $\rho_{od}^{UB} > \rho(X_{\lab} + X_{P_{i}})$, we have identified a path whose cost is smaller than the current upper bound on the cost of the optimal path, and  $\rho_{od}^{UB}$ can be updated to $ \rho(X_{\lab} + X_{P_{i}})$. Improving the upper bound during the algorithm enables to reduce the number of labels considered, and thus speeds up the algorithm. 

The lists $\ell_{v}$ can be built heuristically during some preprocessing. In our experiments, we built them during the computation of the lower bounds $Z_v^{LB}$ with the \SOTAP{} algorithm.
  
\section{\SRCSPP{}}\label{sec:SRCSPP}

\subsection{Complexity}

The \SRCSPP{} is \NP-hard, since the deterministic case is a special case and is already \NP-hard.

\subsection{Algorithm}

As in Section~\ref{subsec:algo_ssp}, we denote by $X_P$ the random variable $\sum_{a \in P} X_{a}$, where independent copies of $X_{a}$'s appearing several time in the sum are used, and we assume given for each vertex $v$ a random variable $Z_{v}^{LB}$ such that $Z_{v}^{LB} \leqst X_{P}$ for each $v$-$d$ path $P$. We assume moreover given for each vertex $v$ the optimal cost $\pi_v$ of an unconstrained $v$-$d$ path. The costs $(\pi_v)_{v\in V}$ can be efficiently computed using Dijkstra's algorithm. \\

The labeling algorithm for the \SRCSPP{} goes as follows. A label $\lab$ is a triple $(v_{\lab}, Y_{\lab}, c_{\lab})$ where $v_{\lab}$ is a vertex, $Y_{\lab}$ is a random variable, and $c_{\lambda} \in \mathbb{R}$ is a cost. Labels are stored in a queue $L$. Initially, $L$ contains a unique label $\lab_{o} = (o,0,0)$, and $c_{od}^{UB} =1 + \sum_{a\in A}c_{a}$. The algorithm ends when $L$ is empty. While $L$ is not empty, the following operations are repeated:
\begin{itemize}
\item[$\bullet$] Extract a label $\lab$ of $\labList$. 
\item[$\bullet$] If $v_{\lab}=d$  and if the inequalities $\riskMeas(Y_{\lab})\leq \riskMeas_0$ and $c_{\lambda} < c_{od}^{UB}$ hold: update $c_{od}^{UB}$ to $c_{\lab}$.
\item[$\bullet$] Otherwise:
If $\riskMeas(Y_{\lab}+Z_v^{LB}) \leq \riskMeas_0$ and $c_{\lab}+\pi_{v_{\lab}} < c_{od}^{UB} $, {\em expand} label $\lab$: for each arc $(v_{\lab},v)$ in $\delta^+(v_{\lab})$, add a new label $(v,Y_{\lab}+X_{(v_{\lab},v)}, c_{\lab} + c_{(v_{\lab},v)})$ to $\labList$. 
\end{itemize}

In the following theorem, we use the notation ``$\sup \rho > \rho_0$'', which means that there exists a random variable $Y$ such that $\rho(Y) > \rho_0$.

\begin{theo}\label{theo:SRCSPP}
Suppose that for each arc $a$ we have $c_a\neq 0$ or $\P(X_a\neq 0)> 0$ and that $\sup \rho > \rho_{0}$. If $\riskMeas$ is consistent with the usual stochastic order, then the algorithm described above terminates after a finite number of iterations and at the end, either $c_{od}^{UB} = 1 +\sum_{a\in A}c_{a} $ and the \SRCSPP{} admits no feasible solutions, or $c_{od}^{UB}$ is finite and equal to the optimal value of a solution of the \SRCSPP.
\end{theo}

In the case when $\sup \rho \leq \rho_0$, any path is feasible, and the optimal solution of the \SRCSPP{} is the unconstrained shortest path, which can be computed with Dijkstra's algorithm.


As noted in Section~\ref{subsec:algo_ssp}, this algorithm can easily be adapted to return a shortest path itself, and to deal with arcs $a$ such that $c_a=0$ and $X_a=0$ by the same kind as the one presented in Remark~\ref{re:Xa0}.

We do not give the proof of Theorem \ref{theo:SRCSPP} since it is very similar to the one of Theorem~\ref{theo:stoBackward}.  We nevertheless explain three ideas on which the algorithm relies. They are very similar to the ones used in Section~\ref{subsec:algo_ssp} for the \SRCSPP. Again, a label $\lab$ corresponds  to some $o$-$v_{\lab}$ path $P_\lab$ that the algorithm tries to expand in an optimal path
 
The first idea is the following. If an $o$-$v$ path $P$ is a subpath of a feasible $o$-$d$ path, then 
$$\riskMeas\left({X_{P}+Z_{v}^{LB}}\right) \leq \riskMeas_{0}.$$
This inequality provides thus a condition that allows to discard partial path that cannot be completed into feasible paths.

The second idea is the following. Assume that $c_{od}^{UB}$ is an upper bound on the optimal cost of the \SRCSPP. Suppose that $c_{od}^{UB}$ is not tight. If $P$ is an $o$-$v$ path subpath of an optimal solution, then
$$c_P+\pi_v < c_{od}^{UB}.$$ It can be used to discard partial paths that cannot be completed into optimal paths.

The third idea is the following. Each time a feasible $o$-$d$ path $P$ satisfying $c_P < c_{od}^{UB}$ is encountered, we have met a better feasible path and $c_{od}^{UB}$ can be updated to $c_P$.

\subsection{Speeding up the algorithm} The same techniques as the ones described in Section \ref{sec:SSPP} can be used to speed up the algorithm.
\subsubsection{Keys}The cost $c_{\lambda} + \pi_{v_{\lambda}}$,  can be used as a key for the element $v_{\lambda}$. It represents a lower bound on the best cost that could be obtained with a path extending $P_{\lab}$.

\subsubsection{Upper bounds on the complete path} 

Compared to the \SSPP{} of Section~\ref{sec:SSPP}, the only difference in the context  of  the \SRCSPP{}
is that a list of pairs $(X_{P_{i}},c_{P_{i}})$ must be stored for each vertex $v$ instead of a list of sole random variables $X_{P_{i}}$.

\section{Experimental results} \label{sec:experiments}

The algorithms have been coded in \texttt{C++}.  The priority queue was implemented using the data-structures \texttt{map} and \texttt{multimap} of the \texttt{C++} standard library. Experiments were carried out on a MacBook Pro with a 2,5 GHz Intel Core i5 processor and 4 Go 
RAM.  Fast Convolution products were performed using the \texttt{BSD} licensed Fast Fourier Transform library \texttt{KissFFT} \citep{kissFFT}. 

\subsection{Instances}
The algorithms were tested on square grid networks of various sizes and distribution types. The origin is the upper left corner of the grid and the destination is the lower right corner of the grid.  A rough description of these instances is available on Table~\ref{tab:instanceDescription}. The name of each instance indicates the width of the grid, and the distribution. Three types of distributions have been considered:  randomly generated generic distributions, lognormal distributions with randomly generated parameters, and gamma distributions with different sizes. 

Instances can be found on the following webpage:
\begin{quotation}
\texttt{http://cermics.enpc.fr/$\sim$parmenta/shortest\_path/}
\end{quotation}

The remaining of this subsection is devoted to additional information regarding the way the instances were generated.

\subsubsection{Distributions of the $X_a$'s}
We describe the way the distributions of arc travel times $X_a$ are built. For all three distributions, the minimum value $X_a$ can take is uniformly drawn at random between $0$ and $50$. We denote this minimum by $t_0$. 

The generic distributions are then built as follows. The size of the support is then drawn uniformly at random between $1$ and $2t_0$. The quantity $\P(X_a=t_0+t)$ is set to  $\frac{r_t}{\sum_t r_t}$, where $r_t$ is generated by selecting at random some prefixed intervals, in which we draw uniformly at random the value of $r_t$. The prefixed intervals are chosen in a way that enforces strong variations in the value of the variances among the arcs.

Regarding the lognormal (resp. gamma) distributions, a maximum value for the mean $M$ is first selected. For the lognormal distributions, M is set to $2t_0$, except for instance g100Ll, for which it is set to $4t_0$. For the gamma distributions, we set $M=10$. Then, a mean $\mu$ is uniformly generated between $1$ and $M$, and a standard deviation $\sigma^2$ is uniformly generated in $[M-\mu, 2M-\mu]$. The quantity $\P(X_a=t_0+t)$ is set to  $\frac{r_t}{\sum_t r_t}$, where $r_t$ is the density of a lognormal (resp. gamma) distributions of mean $\mu$ and variance $\sigma^2$ at time $t$. When $r_t$ becomes smaller than some threshold $\varepsilon$, it is set to $0$. It ensures a finite-size support.

Since we are expecting the complexity of the operations on distributions to be correlated to the size of paths distributions, we provide in Table~\ref{tab:instanceDescription} the size $\ell$ of the \SOTAP{} solution distribution support at the origin.

\subsubsection{Risk measures}
We made the experiments with two risk measures: $\P(\cdot \geq\tau)$ and $CVaR_{\beta}$. 

For the first one, we set $\tau$ to be equal to $\min\{t:\, F_{Z_o}(t)\geq p\}$ for some parameter $p$, where $Z_{o}$ is the solution of the \SOTAP{} at the origin. For the \SSPP{}, $p$ was chosen  in $\{0.5,0.8,0.95\}$. For the \SRCSPP{}, $p$ was set to $0.95$.

We did so because if $\tau$ is too small, then $\P(X_{P}\geq\tau) = 1$ for each path $P$ from origin to destination, which implies that all the solutions of the \SSPP{} have the same solutions, and that the \SRCSPP{} is infeasible. On the contrary, when $\tau$ is too large, all elementary origin-destination paths have the same risk measure $0$, which implies that they are all optimal solutions of the \SSPP{}, and that the \SRCSPP{} is practically an unconstrained deterministic shortest path problem. 

For the second risk measure, we made the experiments for the \SSPP{} with $\beta\in\{0.01, 0.05, 0.25\}$. The values $0.01$ and $0.05$ are usual values of the parameter.  While the value $0.25$ is higher than what is usually used for the $CVaR$, we think that in the context of the \SSPP{}, such a value corresponds to the expected travel time of the worst day in a week, which sounds reasonable when a commuter searches an optimal path. For the \SRCSPP{}, $\beta$ was set to $0.05$. 

\subsubsection{Additional parameters in the case of the \SRCSPP{}}
The \SRCSPP{} requires two additional parameters: the resource constraint $\riskMeas_{0}$ and the arc costs $c_a$. We set $\rho_{0}$ to be equal to $\alpha \rho(Z_{o}) + (1-\alpha)\rho(X_{Q})$ for some $\alpha \in [0,1]$. We use $\rho(Z_{o})$ and $\rho(X_{Q})$, where $Z_{o}$ is again the solution of the \SOTAP{} at the origin, and where $Q$ is the optimal solution of the unconstrained deterministic shortest path problem. For each instance, we solved the \SRCSPP{} for $\alpha\in\{0.02,0.1,0.5\}$.  We did so in order to have feasible but yet nontrivial solutions. The costs $c_a$ were generated randomly, using a uniform law between $1$ and twice the smallest $t$ such that $\P(X_{a}=t) >0$

\begin{table}
\begin{tabular}{|lrrrl|}
\hline
Instance & $|V|$ & $|A|$ & $\ell$  &  Dist. type\\
\hline
\input{instance_description_sept.tex}
\hline 
\end{tabular}
\caption{Grid instances description.}
\label{tab:instanceDescription}
\end{table}

\subsection{\SOTAP{}}
Table~\ref{tab:SOTA} describes the performance of the \SOTAP{} algorithm. ``Upd.'' is the number of updates operations, which can be compared to the number of arcs, and ``Exp.'' is the number of expansion of vertices, which should be compared to the number of vertices in the graph. First note that the algorithm is able to solve the \SOTAP{} on instances with $10'000$ vertices in between $1$ and $10$ seconds. On all the instances we tested, the number of expansions was smaller than $3.3|V|$, and often close to $|V|$, which is the number of expansions needed to check that a solution is a feasible solution of the \SOTAP{}. Finally, the ratio of the number of expansions divided by the number of vertices in the instance increases slowly with the size of the instances, which allows to presume that the performance of the algorithm will remain high on larger instances. The limiting factor in our experiments was the memory available.

\begin{table}
\begin{tabular}{|lrrr|}
\hline
Instance & Upd. & Exp. & CPU time (s) \\ 
\hline
\input{results_SOTA_sept.tex}

\hline 
\end{tabular}
\caption{\SOTAP{} algorithm performances.}
\label{tab:SOTA}
\end{table}

\begin{re}
Note that \SOTAP{} algorithm CPU times to derive bounds to solve \SSPP{} and \SRCSPP{} given in Tables~\ref{tab:SSP_TAU} to~\ref{tab:SRCSP_CVaR} are larger than the those of Table~\ref{tab:SOTA}. Indeed, this is due to the computation of non-dominated path distributions in order to obtain upper bounds as explained in Section~\ref{sec:speedUB}.
\end{re}

\subsection{\SSPP{}}
Table~\ref{tab:SSP_TAU} presents our numerical results when $\riskMeas(\cdot) = \P(\cdot \geq \tau)$ is used, and Table~\ref{tab:SSP_CVAR} when $\riskMeas(\cdot) = CVaR_{\beta}(\cdot)$. In Table~\ref{tab:SSP_TAU}, the value of the parameter $p$ used to choose $\tau$ has been displayed.
The next columns provide the parameters $\tau$ or $\beta$, and the cost of the lower bound obtained by the solution of \SOTAP{} (SOTA) at the origin and the upper bound obtained by applying $\rho(\cdot)$ to the optimal solution of the deterministic shortest path with expectation as cost.  The performance of the \SOTAP{} algorithm used as preprocessing is then given. The next columns indicates the number of labels treated and expanded by the algorithm. Finally, the next columns provide the CPU time for the label algorithm, the cost of the optimal solution, and the total CPU time, which is roughly equal the sum of the CPU time of the \SOTAP{} algorithm and the CPU time of the labeling algorithm. The additional time is due to memory allocation. 

We note that the algorithm is able to solve \SSPP{} for all instances in less that one minute. The most time-consuming phase is always the solution of \SOTAP{} done in preprocessing. The label algorithm runs on most instances in less than one second. The number of labels expanded by the \SSPP{} turns out to be small, which indicates that the quality of the lower bounds provided by the \SOTAP{} algorithm is good. No clear correlations appear between the parameter $\tau$ or $\beta$ and the performance of the algorithm. The cost of the solution of \SOTAP{}, the cost of the optimal solution of \SSPP{}, and the cost of the heuristic solution obtained by taking the path with minimal expectation are within a small range.

We remark that when parameter $p$ is set to be equal to $0.5$, the path with minimum expected length often also minimizes $\P(X\geq\tau)$. This sounds reasonable, because when $p = 0.5$, we expect a path with small $\P(X\geq\tau)$ to have a small median and thus a small expectation.

\begin{scriptsize}
\begin{table}
\begin{tabular}{|l|rrrr|r|rrr|rr|}
\hline
Instance   & $p$ & $\tau $ & SOTA & ESP & SOTA & $\lab$ & $\lab$ & SSPP & Opt. sol. & Total \\
  &&& $\P(\cdot\geq \tau)$  & $\P(\cdot \geq \tau)$ & CPU (s) & treat. & exp. & CPU (s) & $\P(\cdot \geq \tau)$  & CPU (s) \\
 \hline
 \input{results_SSP_Tau_sept_mod.tex}
 \hline 
\end{tabular}
\caption{Performance of \SSPP{} algorithm with $\P(\cdot \geq \tau)$ as risk measure}
\label{tab:SSP_TAU}
\end{table}
\end{scriptsize}

\begin{outdent}
\begin{scriptsize}
\begin{table}
\begin{tabular}{|l|rrr|r|rrr|rr|}
\hline
Instance   & $\beta$ & SOTA & ESP & SOTA& $\lab$ & $\lab$ & SSPP & Opt. sol. & Total \\
  && $CVaR$ & $CVaR$& CPU (s) & treat. & exp. & CPU (s)& $CVaR$ & CPU (s)\\
 \hline
 \input{results_SSP_CVaR_sept_mod.tex}
 \hline 
\end{tabular}
\caption{Performance of \SSPP{} algorithm with $CVaR$ as risk measure}
\label{tab:SSP_CVAR}
\end{table}
\end{scriptsize}
\end{outdent}

\subsection{\SRCSPP{}}

Table~\ref{tab:SRCSP_Tau} contains the numerical experiments for the \SRCSPP{} with $\riskMeas(\cdot) = \P(\cdot \geq \tau)$ as risk measure, and Table~\ref{tab:SRCSP_CVaR} contains the results with $\riskMeas(\cdot) = CVaR_{\beta}(\cdot)$.  Both tables contain the tradeoff $\alpha$ between the risk measure of the solution of \SOTAP{}, the cost of the optimal unconstrained path and the resource constraint $\rho_{0}$. The next column indicates the time consumed by the \SOTAP{} (SOTA) algorithm used in preprocessing.  The three next columns provide the number of labels treated and expanded by the \SRCSPP{} (SRCSP) labeling algorithm and the CPU time of the labeling algorithm. The labeling algorithm was stopped after five minutes. The column $\riskMeas(P)$ is the value of the risk measure evaluated on the solution found, LB is a lower bound on its cost (only displayed when the optimal solution was not found), and $c_P$ is its cost. When a cost equal to $\infty$ is displayed, it means that the algorithm was not able to find any feasible solution. One of the instances was moreover proved to be not feasible (g10G of Table~\ref{tab:SRCSP_CVaR}). Finally, we provide the total CPU time, which is roughly equal to the sum of the preprocessing CPU time and the labeling algorithm CPU time.

The performance of the algorithm is strongly influenced by the parameter $\alpha$, which models the hardness of the resource constraint.  A small $\alpha$ corresponds to a hard resource constraint and a large $\alpha$ corresponds to an easy resource constraint. The algorithm performance relies on the use of stochastic lower bound to cut potentially infeasible path. Thus, when the resource constraint is hard, which corresponds to a small $\alpha$, the number of feasible paths is small, and the algorithm is efficient. The labeling algorithm was able to solve all the instances with $\P(\cdot \geq \tau)$ and most instances with $\riskMeas(\cdot) = CVaR_{\beta}(\cdot)$ when  $\alpha = 0.02$ in the allocated time. On the contrary, when $\alpha$ is larger, the number of feasible paths is much bigger, making the algorithm less efficient.

\begin{scriptsize}
\begin{table}
\begin{outdent}
\begin{tabular}{|l|rrrr|r|rrr|rrrr|r|}
\hline Inst.  & $\tau$ &$\alpha$ & SP &  $\rho_{0}$ &SOTA  & $\lab$ & $\lab$ & SRCSP & $\riskMeas(P)$& LB & $c_{P}$ & gap  & Total  \\
 & & & cost & & CPU (s)&  treat. & exp. &  CPU (s)& & &  & \% & CPU (s)\\
\hline
\input{results_SRCSP_Tau_sept_mod.tex}
\hline
\end{tabular}
\caption{Performance of \SRCSPP{} algorithm with $\P(\cdot \geq \tau)$ as risk measure, with $p = 0.95$ }
\label{tab:SRCSP_Tau}
\end{outdent}
\end{table}
\end{scriptsize}

\begin{scriptsize}
\begin{table}
\begin{outdent}
\begin{tabular}{|l|rrr|r|rrr|rrrr|r|}
\hline Inst.  &$\alpha$ & SP &  $\rho_{0}$ &SOTA  & $\lab$ & $\lab$ & SRCSP & $\riskMeas(P)$& LB & $c_{P}$ & gap  & Total  \\
 & & cost & & CPU (s)&  treat. & exp. &  CPU (s)& & &  & \% & CPU (s)\\
\hline
\input{results_SRCSP_CVaR_sept_mod.tex}
\hline
\end{tabular}
\caption{Performance of \SRCSPP{} algorithm with $CVaR_{\beta}$ with $\beta = 0.05$ as risk measure }
\label{tab:SRCSP_CVaR}
\end{outdent}
\end{table}
\end{scriptsize}


\bigskip


\bibliographystyle{authordate1}

\bibliography{biblioShortestPath}

\end{document}

%% file: instance_description_sept.tex
g10R &  100 &  360 &  191 & Generic\\ 
g40R &  1600 &  6240 &  953& Generic\\ 
g100R &  10000 &  39600 &  2449& Generic\\ 
g10Ls &  100 &  360 &  138& Lognormal\\ 
g40Ls &  1600 &  6240 &  389& Lognormal\\ 
g100Ls &  10000 &  39600 &  797& Lognormal\\ 
g100Ll &  10000 &  39600 &  2091& Lognormal - long \\ 
g10G &  100 &  360 &  157& Gamma\\ 
g40G &  1600 &  6240 &  538& Gamma\\ 
g100G & 10000 &  39600 & 1301 & Gamma \\

%% file: results_SOTA_sept.tex
g10R &  398 &  111 &  0.0037 \\ 
g40R &  14060 &  3598 &  0.3944  \\ 
g100R &  103397 &  26095 &  6.0776  \\ 
g10Ls &  403 &  113 &  0.0061 \\ 
g40Ls &  11074 &  2838 &  0.2712  \\ 
g100Ls &  80291 &  20271 &  3.1419  \\ 
g100Ll &  129797 &  32764 &  14.4251  \\ 
g10G &  430 &  121 &  0.0053 \\ 
g40G &  13696 &  3513 &  0.2853 \\ 
g100G & 107802 & 27214 & 4.1233 \\ 

%% file: results_SSP_Tau_sept_mod.tex
g40R &    0.500 & 1773 &  0.495 &  0.498 &  0.788 & 304 & 79 &  0.081 &  0.498 &  0.878 \\
g40R &    0.800 & 1817 &  0.197 &  0.218 &  0.772 & 304 & 82 &  0.055 &  0.203 &  0.836 \\
g40R &    0.950 & 1855 &  0.049 &  0.074 &  0.784 & 304 & 83 &  0.059 &  0.051 &  0.852 \\
\hline g100R &   0.500 & 4492 &  0.497 &  0.502 &  12.153 & 794 & 201 &  0.484 &  0.502 &  12.695 \\
g100R &   0.800 & 4565 &  0.197 &  0.203 &  12.110 & 1466 & 370 &  0.913 &  0.200 &  13.080 \\
g100R &   0.950 & 4630 &  0.049 &  0.057 &  12.729 & 794 & 214 &  0.498 &  0.050 &  13.287 \\
\hline g40Ls &  0.500 & 1242 &  0.488 &  0.502 &  0.550 & 304 & 79 &  0.031 &  0.502 &  0.588 \\
g40Ls &   0.800 & 1268 &  0.199 &  0.217 &  0.544 & 452 & 117 &  0.048 &  0.217 &  0.598 \\
g40Ls &   0.950 & 1297 &  0.047 &  0.061 &  0.544 & 304 & 81 &  0.030 &  0.051 &  0.582 \\
\hline g100Ls &  0.500 & 2959 &  0.493 &  0.496 &  6.203 & 794 & 201 &  0.132 &  0.496 &  6.378 \\
g100Ls &  0.800 & 2999 &  0.197 &  0.203 &  6.254 & 794 & 201 &  0.140 &  0.203 &  6.436 \\
g100Ls &  0.950 & 3040 &  0.049 &  0.053 &  6.322 & 1082 & 273 &  0.176 &  0.050 &  6.542 \\
\hline g100Ll &  0.500 & 3971 &  0.498 &  0.503 &  30.104 & 1765 & 444 &  1.339 &  0.503 &  31.512 \\
g100Ll &  0.800 & 4070 &  0.200 &  0.204 &  30.186 & 1260 & 320 &  0.936 &  0.203 &  31.191 \\
g100Ll &  0.950 & 4170 &  0.050 &  0.053 &  30.156 & 1577 & 407 &  1.162 &  0.051 &  31.387 \\ 
\hline g40G &   0.500 & 658 &  0.485 &  0.485 &  0.539 & 301 & 79 &  0.052 &  0.485 &  0.599 \\
g40G &   0.800 & 680 &  0.190 &  0.192 &  0.535 & 301 & 79 &  0.048 &  0.192 &  0.590 \\
g40G &   0.950 & 701 &  0.050 &  0.052 &  0.539 & 301 & 81 &  0.048 &  0.050 &  0.594 \\
\hline g100G& 0.500& 1708& 0.492& 0.498& 8.702& 2679& 673& 0.928& 0.498& 9.679 \\ 
g100G& 0.800& 1746& 0.195& 0.204& 8.601& 8256& 2079& 2.906& 0.202& 11.555 \\
g100G& 0.950& 1782& 0.049& 0.058& 8.624& 8436& 2139& 2.992& 0.052& 11.666 \\

%% file: results_SSP_CVaR_sept_mod.tex
g40R&  0.250&  1837.192&  1845.262&  0.800& 260& 83&  0.052&  1838.437&  0.860\\
g40R&  0.050&  1873.505&  1889.451&  0.776& 258& 87&  0.053&  1874.776&  0.837\\
g40R&  0.010&  1899.092&  1922.075&  0.775& 253& 94&  0.055&  1901.063&  0.839\\
\hline g100R&  0.250&  4599.412&  4603.778&  12.871& 784& 211&  0.495&  4600.425&  13.424\\
g100R&  0.050&  4663.326&  4672.600&  12.481& 1076& 303&  0.715&  4664.403&  13.254\\
g100R&  0.010&  4711.355&  4724.452&  12.156& 1076& 315&  0.715&  4712.107&  12.929\\
\hline g40Ls&  0.250&  1283.573&  1287.441&  0.545& 282& 81&  0.030&  1285.220&  0.582\\
g40Ls&  0.050&  1312.482&  1319.501&  0.543& 273& 86&  0.031&  1313.445&  0.581\\
g40Ls&  0.010&  1336.079&  1345.904&  0.543& 259& 93&  0.037&  1336.710&  0.587\\
\hline g100Ls&  0.250&  3021.004&  3022.701&  6.274& 731& 203&  0.125&  3021.967&  6.441\\
g100Ls&  0.050&  3062.436&  3065.510&  6.323& 934& 276&  0.179&  3063.256&  6.560\\
g100Ls&  0.010&  3095.327&  3099.824&  6.265& 612& 205&  0.116&  3095.587&  6.425\\
\hline g100Ll&  0.250&  4124.031&  4126.803&  30.856& 1269& 335&  0.959&  4125.520&  31.885\\
g100Ll&  0.050&  4224.099&  4228.733&  30.691& 755& 211&  0.565&  4225.118&  31.325\\
g100Ll&  0.010&  4302.636&  4309.323&  30.432& 737& 214&  0.556&  4303.215&  31.057  \\ 
\hline g40G&  0.250&  691.308&  691.725&  0.554& 272& 80&  0.055&  691.523&  0.619\\
g40G&  0.050&  713.333&  714.331&  0.551& 252& 82&  0.045&  713.517&  0.604\\
g40G&  0.010&  730.689&  732.304&  0.555& 224& 85&  0.044&  730.875&  0.607\\
\hline g100G& 0.250& 1765.151& 1768.065& 8.640& 8051& 2083& 2.904& 1766.616& 11.592 \\
g100G& 0.050& 1801.330& 1807.005& 8.651& 12803& 3301& 4.625& 1803.171& 13.326 \\ 
g100G& 0.010& 1829.399& 1837.500& 8.605& 19187& 4952& 6.975& 1831.498& 15.628 \\ 

%% file: results_SRCSP_Tau_sept_mod.tex
g10R & 498 &  0.02 &  134.106 &  0.069 &  0.008 & 67 & 19 &  0.003 &  0.050 &   &  320.702 & 0.0  &  0.012 \\
g10R & 498 &  0.10 &  134.106 &  0.145 &  0.008 & 79 & 22 &  0.003 &  0.070 &   &  219.858 &  0.0 &  0.012\\
g10R & 498 &  0.50 &  134.106 &  0.525 &  0.007 & 318 & 82 &  0.013 &  0.258 &   &  202.102 & 0.0 &  0.021\\
\hline g40R & 1855 &  0.02 &  584.055 &  0.068 &  0.733 & 1827 & 462 &  0.319 &  0.062 &   &  1283.772 & 0.0 &  1.063\\
g40R & 1855 &  0.10 &  584.055 &  0.144 &  0.724 & 75110 & 19177 &  13.486 &  0.123 &   &  1205.852 & 0.0 &  14.220\\
g40R &  1855 &  0.50 &  584.055 &  0.525 &  0.726 &  1355571 &  429418 &  300.000 &  0.479  &  856.668 &  1248.012 &  45.7 &  300.735 \\
\hline g10Ls & 418 &  0.02 &  251.844 &  0.069 &  0.012 & 111 & 30 &  0.005 &  0.059 &   &  421.254 & 0.0 &  0.018\\
g10Ls & 418 &  0.10 &  251.844 &  0.145 &  0.012 & 360 & 94 &  0.016 &  0.146 &   &  400.475 & 0.0 &  0.029\\
g10Ls & 418 &  0.50 &  251.844 &  0.525 &  0.013 & 1662 & 427 &  0.078 &  0.516 &   &  371.532 & 0.0 &  0.092\\
\hline g40Ls & 1297 &  0.02 &  833.451 &  0.066 &  0.533 & 1147 & 290 &  0.103 &  0.061 &   &  1541.857 & 0.0 &  0.646\\
g40Ls & 1297 &  0.10 &  833.451 &  0.143 &  0.531 & 119204 & 30376 &  10.991 &  0.120 &   &  1490.817 & 0.0 &  11.533\\
g40Ls &  1297 &  0.50 &  833.451 &  0.524 &  0.547 &  2823376 &  861653 &  300.000 &  0.480  &  1220.566 &  1433.313 &  17.4 &  300.556 \\
\hline g10G & 185 &  0.02 &  213.270 &  0.062 &  0.009 & 117 & 31 &  0.005 &  0.058 &   &  297.191 & 0.0 &  0.014\\
g10G & 185 &  0.10 &  213.270 &  0.116 &  0.010 & 175 & 46 &  0.007 &  0.076 &   &  269.665 & 0.0 &  0.018\\
g10G & 185 &  0.50 &  213.270 &  0.386 &  0.008 & 149 & 45 &  0.007 &  0.243 &   &  231.777 & 0.0 &  0.016\\
\hline g40G & 701 &  0.02 &  892.028 &  0.069 &  0.535 & 10528 & 2635 &  1.758 &  0.064 &   &  1577.910 & 0.0 &  2.305\\
g40G & 701 &  0.10 &  892.028 &  0.145 &  0.527 & 1231773 & 307952 &  207.613 &  0.143 &   &  1473.075 & 0.0 &  208.151 \\
g40G & 701 &  0.50 &  892.028 &  0.525 &  0.553 & 1468578 & 551686 &  300.000 & - &  999.764 &  $\infty$ &  $\infty$ &  300.560\\ 

%% file: results_SRCSP_CVaR_sept_mod.tex
g10R &  0.02 &  134.106 &  509.474 &  0.007 & 67 & 19 &  0.002 &  505.811 &   &  219.858 & 0.0   &  0.010\\
g10R &  0.10 &  134.106 &  533.220 &  0.007 & 382 & 98 &  0.016 &  513.358 &   &  202.102 & 0.0  &  0.025\\
g10R &  0.50 &  134.106 &  651.949 &  0.008 & 612 & 244 &  0.035 &  645.537 &   &  150.467 &  0.0 &  0.044\\
\hline g40R &  0.02 &  584.055 &  1896.338 &  0.765 & 28474 & 7235 &  4.913 &  1889.475 &  &  1205.852 &  0.0 &  5.691\\
g40R &  0.10 &  584.055 &  1987.669&  0.727  &  1256561 &  427909 &  300.000 &  1987.238  &  793.711 &  1050.238 &  32.3 &  300.736 \\
g40R &  0.50 &  584.055 &  2444.322&  0.731  &  670441 &  514178 &  300.000 &  2439.891  &  613.004 &  742.644 &  21.1 &  300.741 \\
\hline g10Ls &  0.02 &  251.844 &  432.411 &  0.012 & 111 & 30 &  0.005 &  429.405 &   &  421.254 & 0.0  &  0.018\\
g10Ls &  0.10 &  251.844 &  453.657 &  0.013 & 1050 & 268 &  0.049 &  445.882 &   &  374.680 &  0.0 &  0.064\\
g10Ls &  0.50 &  251.844 &  559.885 &  0.012 & 1357 & 353 &  0.076 &  513.708 &   &  283.292 &  0.0 &  0.089\\
\hline g40Ls &  0.02 &  833.451 &  1334.621 &  0.521 & 94912 & 24051 &  8.796 &  1333.675 &   &  1490.817 & 0.0  &  9.328\\
g40Ls &  0.10 &  833.451 &  1423.180&  0.558  &  2520108 &  1025174 &  300.000 &  1417.246  &  1028.961 &  1336.761 &  29.9 &  300.569 \\
g40Ls &  0.50 &  833.451 &  1865.975&  0.538  &  1202558 &  836889 &  300.001 &  1860.527  &  880.748 &  1007.203 &  14.4 &  300.548 \\
\hline g10G &  0.02 &  213.270 &  192.069 &  0.009 & 101 & 26 &  0.004 &   &   & INFEAS. &   & 0.014\\
g10G &  0.10 &  213.270 &  195.236 &  0.009 & 118 & 31 &  0.005 &  194.351 &   &  269.665 & 0.0  &  0.015\\
g10G &  0.50 &  213.270 &  211.069 &  0.009 & 173 & 50 &  0.008 &  203.028 &   &  231.777 &  0.0 &  0.018\\
\hline g40G &  0.02 &  892.028 &  719.158 &  0.547 & 18664 & 4669 &  3.133 &  718.039 &   &  1562.206 &  0.0 &  3.690\\
g40G &  0.10 &  892.028 &  742.458 &  0.562 &  1631484 &  505573 &  300.000 &  -  &  1121.193 &  $\infty$ &  $\infty$ &  300.573 \\
g40G &  0.50 &  892.028 &  858.958 &  0.573 &  839101 &  632986 &  300.000 &  857.531  &  923.304 &  1163.971 &  26.1 &  300.583 \\